\documentclass[11pt]{article}
\usepackage[]{amsmath,amssymb,amsthm,bm}
\usepackage[]{fullpage}
\usepackage[noend]{algorithmic}
\usepackage{algorithm}
\usepackage[colorlinks=true]{hyperref}
\usepackage{xcolor}
\title{\bf A New Approximation Guarantee for Monotone Submodular Function Maximization via Discrete Convexity}
\author{Tasuku Soma\footnote{Supported by JSPS Grant-in-Aid for Young Scientists (Start-up) Grand Number JP16H06676.}\\
The University of Tokyo\\
\texttt{tasuku\_soma@mist.i.u-tokyo.ac.jp}
\and
Yuichi Yoshida\footnote{Supported by JST ERATO Grant Number JPMJER1305 and JSPS KAKENHI Grant Number JP17H04676}\\
National Institute of Informatics\\
\texttt{yyoshida@nii.ac.jp}
}

\newcommand{\bbR}{\mathbb{R}}
\newcommand{\bbZ}{\mathbb{Z}}
\newcommand{\caL}{\mathcal{L}}
\newcommand{\bd}{{\bm d}}
\newcommand{\bp}{{\bm p}}
\newcommand{\bw}{{\bm w}}
\newcommand{\bx}{{\bm x}}
\newcommand{\by}{{\bm y}}

\newcommand{\bv}{{\bm v}}
\newcommand{\bM}{\mathbf{M}}
\newcommand{\caI}{\mathcal{I}}
\newcommand{\bfone}{\mathbf{1}}
\newcommand{\bfzero}{\mathbf{0}}
\newcommand{\M}{M$^\natural$}
\renewcommand{\emptyset}{\varnothing}
\newcommand{\set}[1]{\{#1\}}

\DeclareMathOperator{\E}{\mathbf{E}}
\DeclareMathOperator{\argmin}{argmin}

\DeclareMathOperator{\Hess}{Hess}
\DeclareMathOperator{\Diag}{Diag}

\newtheorem{theorem}{Theorem}
\newtheorem{lemma}{Lemma}

\theoremstyle{definition}
\newtheorem{definition}{Definition}
\newtheorem{remark}{Remark}
\newtheorem{example}{Example}
\begin{document}

\maketitle
\begin{abstract}
    In monotone submodular function maximization, approximation guarantees based on the curvature of the objective function have been extensively studied in the literature.
    However, the notion of curvature is often pessimistic, and we rarely obtain improved approximation guarantees, even for very simple objective functions.

    In this paper, we provide a novel approximation guarantee by extracting an \M-concave function $h:2^E \to \bbR_+$, a notion in discrete convex analysis, from the objective function $f:2^E \to \bbR_+$.
    We introduce the notion of $h$-curvature, which measures how much $f$ deviates from $h$, and show that we can obtain a $(1-\gamma/e-\epsilon)$-approximation to the problem of maximizing $f$ under a cardinality constraint in polynomial time for any constant $\epsilon > 0$.
    Then, we show that we can obtain nontrivial approximation guarantees for various problems by applying the proposed algorithm.
\end{abstract}

\thispagestyle{empty}
\setcounter{page}{0}

\newpage


\section{Introduction}
A set function $f:2^E \to \bbR$ is called \emph{monotone} if $f(X) \leq f(Y)$ for any $X \subseteq Y \subseteq E$ and called \emph{submodular} if $f(X)+f(Y) \geq f(X \cap Y)+f(X \cup Y)$ for any $X,Y \subseteq E$.
In monotone submodular function maximization under a cardinality constraint, given a nonnegative monotone submodular function $f:2^E \to \bbR_+$ and $k \in \bbZ_+$, we want to compute a set $S \subseteq E$ with $|S| \leq k$ that maximizes $f(S)$.
This simple problem includes various optimization problems in theory and practice, such as facility location, combinatorial auctions~\cite{Dobzinski:2010dt,Lehmann:2006kb}, viral marketing in social networks~\cite{Kempe:2003iu}, and sensor placement~\cite{Krause:2008vo,Krause:2011fz}.

Although monotone submodular function maximization is an NP-hard problem, the greedy algorithm achieves a $(1-1/e)$-approximation, and this approximation ratio is known to be tight~\cite{Nemhauser:1978dm}.
In practice, however, the greedy algorithm exhibits an approximation ratio much better than $1-1/e$, sometimes close to one~\cite{Krause:2008vo}.
This gap between the theoretical guarantee and the practical performance in ``real-world'' instances has been a major mystery in submodular function maximization.
The first attempt to explain this gap goes back to Conforti and G{\'e}rard~\cite{Conforti:1984ig}.
They introduced a parameter called the (total) \emph{curvature} of a monotone submodular function, with which we can derive a tighter approximation ratio.
\begin{definition}[Curvature]
  The curvature $c$ of a monotone submodular function $f: 2^E \to \bbR_+$ is
  \[
    c := 1 - \min_{i \in E} \frac{f(i \mid E - i)}{f(i)}.
  \]
\end{definition}
Here, $f(i)$ and $f(i \mid E-i)$ are shorthand for $f(\set{i})$ and $f(E)-f(E-i)$, respectively.
We note that $f(i \mid E -i) \leq f(i)$ holds from the submodularity of $f$; hence, $c \in [0,1]$ holds.
Roughly speaking, the curvature measures how close it is to a modular function, where a set function $g:2^E \to \bbR$ is called \emph{modular} if $g(X)+g(Y)=g(X\cap Y)+g(X\cup Y)$ for every $X,Y \subseteq E$.
Indeed, we can easily observe that $c = 0$ if and only if $f$ is modular.

It is shown in~\cite{Conforti:1984ig} that the greedy method achieves a $(1-e^{-c})/c$-approximation, which is at least $1-1/e$ and tends to $1$ as $c \to 0$.
The approximation has recently improved to $1-c/e-\epsilon$ for any $\epsilon > 0$ by using a more sophisticated algorithm by Sviridenko, Vondr\'{a}k, and Ward~\cite{Sviridenko2015}.
Since the curvature is easy to analyze, it has been shown that we can obtain refined approximation guarantees for various settings by exploiting the curvature~\cite{Iyer2013,Balkanski2016a}.

However, we might ask \emph{does the curvature explain the gap completely?}
The answer seems to be negative.
The concept of curvature is still unsatisfactory because monotone submodular functions in practical applications are often far from modular functions, and the curvature does not explain why we can obtain high approximation ratios.
For instance, let us consider a very simple function $f(X) = \sqrt{|X|}$.
Since this function has a curvature of $1-O(1/\sqrt{n})$ for $n = |E|$, the approximation guarantee is (roughly) $1-1/e$, while the greedy algorithm obviously finds an optimal solution!

\subsection{Our contributions}

To narrow the gap discussed above, we consider a larger and richer class of functions that are easy to maximize in this work.
For such a class of functions, we exploit \M-concave functions, introduced in the discrete convex analysis literature~\cite{Murota:1999fi}.
A set function $f : 2^E \to \bbR$ is called \emph{\M-concave}\footnote{For a set function, \M-concave functions are essentially equivalent to \emph{valuated matroids}.} if, for any $X, Y \subseteq E$ and $i \in X \setminus Y$, either
(i) $f(X) + f(Y) \leq f(X - i) + f(Y + i)$
or
(ii) there exists $j \in Y \setminus X$ such that $f(X) + f(Y) \leq f(X - i + j) + f(Y + i - j)$.
Intuitively speaking, we can increase the sum $f(X)+f(Y)$ by making $X$ and $Y$ closer, which resembles concave functions.
We note that an \M-concave function is submodular, and an \M-concave function can be maximized in polynomial time with the greedy algorithm (see, e.g.,~\cite{Murota2003}).
Hence, \M-concave functions can be regarded as a class of ``easy'' submodular functions.

Now, we start with the following observation.
For a monotone submodular function $f: 2^E \to \bbR_+$, we define a modular function $h$ as $h(X):= \sum_{i \in X} f(i \mid E \setminus i)$.
Then, $f$ can be decomposed as $f = g + h$, where $g$ is another monotone submodular function.
One can show that $h(X) \geq (1-c)f(X)$ for any $X \subseteq E$, where $c$ is the curvature of $f$.
The modular function $h$ can be regarded as an ``easy'' part of $f$ and $g$ as a ``difficult'' part, and the curvature measures the contribution of the modular part $h$ to the entire function $f$.
Observing that any modular function is \M-concave, we can think of decompositions into a monotone submodular function and an \M-concave function, which gives the following definition:
\begin{definition}[$h$-curvature]
  Suppose that a monotone submodular function $f:2^E \to \bbR_+$ is decomposed as $g+h$, where $g:2^E \to \bbR_+$ is monotone submodular, and $h:2^E \to \bbR_+$ is \M-concave.
  Then, the \emph{$h$-curvature} $\gamma_h$ of $f$ against $h$ is the minimum value of $\gamma$ such that $h(X) \geq (1-\gamma) f(X)$ for every $X \subseteq E$, or equivalently
  \[
    \gamma_h = 1 - \min_{X \subseteq E}\frac{h(X)}{f(X)} = \max_{X \subseteq E}\frac{g(X)}{f(X)},
  \]
  where we conventionally assume that $\frac{0}{0} := \infty$.
\end{definition}
We note that the $h$-curvature is defined for a decomposition $f=g+h$, whereas the standard curvature is defined for the function $f$ itself.
By the argument above, we can always find a decomposition $f=g+h$ such that the $h$-curvature is no greater than the standard curvature.

Our main contribution is the following generalization of the results of~\cite{Sviridenko2015,Murota2003} in terms of the $h$-curvature.
\begin{theorem}\label{thm:main}
  Let $f:2^E \to \bbR_+$ be a monotone submodular function that can be decomposed as $f=g+h$, where $g$ is monotone submodular, and $h$ is \M-concave.
  Assume that we have value oracles of $g$ and $h$.
  Then, for an integer $k \in \bbZ_+$ and a constant $\epsilon > 0$, we can find a random subset $X$ of size $k$ such that $\E[f(X)] \geq (1-\gamma_h / e-\epsilon)f(O)$ in polynomial time, where $O$ is an optimal solution.
\end{theorem}

\paragraph{Applications}
\M-concave functions include many nontrivial functions such as (weighted) matroid rank functions and laminar concave functions, which are far from modular functions.
Our result immediately implies that if the significant part of $f = g + h$ is due to the \M-concave part $h$, we can obtain a better approximation ratio, although it might not be straightforward to find such a decomposition of $f$ into $g$ and $h$ even if $f$ is a compactly represented function.
In Section~\ref{sec:decomposition}, we provide a general algorithm for finding a decomposition from the value oracle of $f$.
This algorithm always finds a decomposition that is at least as good as the decomposition via the standard curvature.
In Section~\ref{sec:applications}, we provide problem-specific decompositions for several problems such as facility location, which yield improved approximation guarantees.

\subsection{Proof technique}
Our main result (Theorem~\ref{thm:main}) is proved by modifying the \emph{continuous greedy algorithm}.
The continuous greedy algorithm is a powerful and flexible framework for submodular function maximization, which has been applied to a matroid constraint, a knapsack constraint, and even combinations of various constraints~\cite{Calinescu2011,Chekuri2010,Vondrak2014}.
At a high level, the continuous greedy algorithm generates a sequence $\bx(t)$ in the convex hull of feasible solutions for $t \in [0,1]$ by the following differential equation: $\frac{d\bx}{dt} = \bv(t)$ and $\bx(0) = \bfzero$, where $\bv$ is a prescribed velocity vector.
Then, we round the final point $\bx(1)$ into a feasible solution $X$ via a rounding algorithm.
The previous result of \cite{Sviridenko2015} relies on the property that for a modular function $h = \bw(X)$, if $\bw^\top\bv(t) \geq \alpha$ holds for $t \in [0,1]$, then $\bw^\top \bx(1) \geq \alpha$.
This property enables us to find a fractional solution $\bx(1)$ that simultaneously optimizes $h$ and $g$;
that is, $\bx(1)$ achieves the optimal value for the modular part $h$ and a $(1-1/e)$-approximation for the remaining monotone submodular part $g$.

\paragraph{Continuous greedy algorithm with a multilinear extension and concave closure}
To run the continuous greedy algorithm, we require continuous extensions of set functions.
For submodular functions, the \emph{multilinear extension} is typically used.
In addition, a modular function is trivially extendable.
However, for \M-concave functions, it is nontrivial to choose a continuous extension because if we just use a continuous extension, we end up with a $(1-1/e)$-approximation.
To remedy this, we use a different continuous extension, namely, the \emph{concave closure}.
The concave closure is difficult to evaluate in general, but for \M-concave functions, we can evaluate their concave closures in polynomial time thanks to the results of discrete convex analysis.
This property enables us to run a modified continuous greedy algorithm with the following strong guarantee;
that is, our continuous greedy algorithm finds a fractional solution that achieves the optimal value for the \M-concave part and a $(1-1/e)$-approximation for the remaining monotone submodular part.

\paragraph{Rounding algorithm preserving the values of a multilinear extension and concave closure}
A difficulty also arises in the rounding phase.
Typically, a rounding algorithm finds a feasible subset that preserves the value of the multilinear extension.
However, since our modified continuous greedy algorithm also involves a concave closure, we need to design a rounding algorithm that preserves the values of the multilinear extension and concave closure.
To this end, we extend the \emph{swap rounding}~\cite{Chekuri2010} algorithm for \M-concave functions.
The original swap rounding algorithm is designed for rounding a fractional solution in a matroid base polytope into a matroid base without losing the value of a multilinear extension.
We prove that almost the same strategy works for our purpose by exploiting the combinatorial structures of \M-concave functions.

\subsection{Related work}
The concept of curvature was introduced by Conforti and Cornu{\'e}jols, and they proved that the classical greedy algorithm of \cite{Nemhauser:1978dm} achieves a $(1-e^{-c})/c$-approximation for a cardinality constraint.
In \cite{Vondrak2010a}, Vondr\'{a}k introduced a slightly relaxed variant of the curvature, namely, the \emph{curvature with respect to the optimum}.
He obtained the same approximation guarantee $(1-e^{-c^*})/c^*$ for a matroid constraint, where $c^*$ is the curvature with respect to the optimum, and also showed that the approximation guarantee is tight for general submodular functions.
Later, this result was refined with the total curvature by Sviridenko, Vondr\'{a}k, and Ward~\cite{Sviridenko2015}.
For a knapsack constraint, Yoshida~\cite{Yoshida2016} proved a better approximation ratio of $1-c/e$.
Iyer and Bilmes~\cite{Iyer2013} studied a general reduction between submodular function maximization and a \emph{submodular cover}, and they proved that if the curvature of the functions involved is small, the two problems reduce to each other.

Discrete convex analysis originated with Murota~\cite{Murota2003}.
We note that discrete convex functions are usually defined on the integer lattice $\bbZ^E$, although we focus on set functions in this paper.
Discrete convex functions admit various attractive properties such as Fenchel duality, the separation theorem, and descent-type algorithms, which are analogous to convex analysis in the Euclidean space.
The applications of discrete convex analysis include economics, system analysis for electrical circuits, phylogenetic analysis, etc.
Shioura~\cite{Shioura2009} studied maximization of the sum of \M-concave set functions subject to a matroid constraint.
He showed that pipage rounding~\cite{Calinescu2011} can be explained from the viewpoint of discrete convexity.
For details, the reader is referred to a monograph~\cite{Murota2003} or recent survey~\cite{Murota2016}.

Recently, Chatziafratis, Roughgarden, and Vondr\'{a}k~\cite{Chatziafratis2017} provided improved approximation guarantees via a different concept called \emph{perturbation stability}.
They proved that the greedy and local search algorithms achieve a better approximation for submodular function maximization for various set systems under a stability assumption.

\subsection{Organization}
In Section~\ref{sec:pre}, we introduce the notation and basic concepts for submodular function maximization and discrete convex analysis.
Our main theorem and continuous greedy algorithms are presented in Section~\ref{sec:alg}.
We describe a general algorithm for finding a decomposition of a given monotone submodular function into a monotone submodular function and an \M-concave function in Section~\ref{sec:decomposition}.
In Section~\ref{sec:applications}, we provide several examples of problem-specific decompositions and the theoretical bounds of the $h$-curvature.



\section{Preliminaries}\label{sec:pre}

For a set $S \subseteq E$, $\bfone_S$ denotes the characteristic vector of $S$; that is, $\bfone_S(i) = 1$ if $i \in S$, and $\bfone_S(i) = 0$ otherwise.
The dimension of the ambient space should be clear from the context.
For $a \in \bbR$, we define $[a]^+ := \max\set{a, 0}$.
For a vector $\bx \in \bbR^E$ and $X \subseteq E$, we use the shorthand notation: $\bx(X) := \sum_{i \in X}\bx(i)$.

A pair consisting of a finite set $E$ and a set family of $\caI \subseteq 2^E$ is called a \emph{matroid} if (i) $\emptyset \in \caI$; (ii) if $X \in \caI$, then any subset of $X$ also belongs to $\caI$; and (iii) if $X, Y \in \caI$ and $|X| < |Y|$, then there exists $i \in Y \setminus X$ such that $X + i \in \caI$.
A member of $\caI$ is called an \emph{independent set}.
A maximal independent set is called a \emph{base}.
The \emph{base polytope} of a matroid is the convex hull of all characteristic vectors of its bases.
The \emph{rank function} of a matroid is the following set function: $r(X) = \max\{|I| : I \subseteq X, I \in \caI\}$.
Although it is well-known that matroid rank functions are monotone and submodular, they are indeed included in the more tractable class called \M-concave functions.

\begin{definition}[\M-concave function~\cite{Murota2000}]
    A set function $f: 2^E \to \bbR$ is \emph{\M-concave} if for $X, Y \subseteq E$ and $i \in X \setminus Y$, either
    \begin{align}\label{eq:mconcave-incr}
        f(X) + f(Y) \leq f(X - i) + f(Y + i),
    \end{align}
    or there exists $j \in Y \setminus X$ such that
    \begin{align}\label{eq:mconcave-exc}
        f(X) + f(Y) \leq f(X - i + j) + f(Y - i + j).
    \end{align}
\end{definition}

It is known that \M-concave functions are submodular.
Note that \M-concave functions are not necessarily monotone.
We give several examples of \M-concave functions below.

\begin{example}[\M-concave functions from univariate concave functions]
    Let $\phi: \bbR_+ \to \bbR$ be a concave function.
    Then, $f(X) := \phi(|X|)$ is an \M-concave function.
    For example, $X \mapsto \sqrt{|X|}$ and  $X \mapsto \min\{|X|, \alpha\}$ ($\alpha > 0$) are \M-concave.
    More generally, let $\caL$ be a laminar family on $E$ and $\phi_L : \bbR_+ \to \bbR$ be a concave function for $L \in \caL$.
    Then, $f(X) := \sum_{L \in \caL} \phi_L(|X \cap L|)$ is an \M-concave function.
    Here, as a special case, a separable concave function $f(X) = \sum_{i \in E} \phi_i(|X \cap \{i\}|)$ is \M-concave.
\end{example}

\begin{example}[\M-concave functions from matroids]\label{ex:weighted-rank}
    Let $\bM = (E, \caI)$ be a matroid and $\bw \in \bbR_+^E$.
    The \emph{weighted matroid rank function} is the following set function:
    \begin{align}
        f(X) = \max\{\bw(I) : I \subseteq X, I \in \caI\}.
    \end{align}
    That is, $f(X)$ is the maximum weight of independent sets in $X$.
    In particular, the matroid rank function is \M-concave.
\end{example}

If we restrict the domain of an \M-concave function to all subsets of the same cardinality, the condition in \eqref{eq:mconcave-incr} can be omitted.
Such functions are called \emph{M-concave functions}.

For exploiting continuous greedy algorithms, we require continuous extensions of set functions.
\begin{definition}[Multilinear extension]
  The \emph{multilinear extension} $F:[0,1]^E \to \bbR$ of a set function $f:2^E \to \bbR$ is defined as
  \[
    F(\bx) = \sum_{X \subseteq E}f(X) \prod_{i \in X}\bx(i) \prod_{i \in E \setminus X}(1-\bx(i)).
  \]
\end{definition}

\begin{definition}[Concave closure]
    The \emph{concave closure} $\bar f:[0,1]^E \to \bbR$ of a set function $f : 2^E \to \bbR$ is defined as
    \begin{align}
        \bar f (\bx) := \max \left\{ \sum_{Y \subseteq E} \lambda_Y f(Y) : \sum_{Y} \lambda_Y \bfone_Y = \bx, \sum_Y \lambda_Y = 1, \lambda_Y \geq 0 \, (Y \subseteq E) \right\}.
    \end{align}
\end{definition}

If $f$ is an \M-concave function, then $\bar f(\bx)$ can be computed in polynomial time for any $\bx \in [0,1]^E$~\cite{Shioura2015}.
Furthermore, the subgradients of $\bar f$ can be computed in polynomial time~\cite{Shioura2015}. 
This yields a separation oracle for the constraint of the form $\bar f(\bx) \geq \alpha$.


\section{Algorithms}\label{sec:alg}

In this section, we prove Theorem~\ref{thm:main}.
To this end, we show the following, which simultaneously optimizes $g$ and $h$:
\begin{theorem}\label{thm:simultaneous-approximation}
    Let $f:2^E \to \bbR_+$ be a monotone submodular function decomposed as $f=g+h$ for a monotone submodular function $g:2^E \to \bbR_+$ and an \M-concave function $h:2^E\to \bbR$, and let $k$ be a positive integer.
    Then, there exists a polynomial-time algorithm that finds a random set $X \subseteq E$ of cardinality $k$ such that
    \begin{align*}
        \E[g(X)] &\geq \Bigl(1-\frac{1}{e}\Bigr) g(O) - \epsilon M, \\
        \E[h(X)] &\geq h(O) - \epsilon M,
    \end{align*}
    for any constant $\epsilon > 0$, where $O$ is an optimal solution for $\max\{ f(X) : |X| \leq k\}$ and $M = \max_{e \in E}\max\set{g(e),h(e)}$.
\end{theorem}
We note that the special case in which $h$ is modular is proved in~\cite{Sviridenko2015}.

Before proving Theorem~\ref{thm:simultaneous-approximation}, we see that Theorem~\ref{thm:main} is an easy consequence of Theorem~\ref{thm:simultaneous-approximation}:
\begin{proof}[Proof of Theorem~\ref{thm:main}]
    By Theorem~\ref{thm:simultaneous-approximation}, we can compute a random $X \subseteq E$ of size $k$ such that
    \begin{align*}
        \E[f(X)]
        &\geq \Bigl(1- \frac{1}{e}\Bigr)g(O) + h(O)-2\epsilon M
        = \Bigl(1-\frac{1}{e}\Bigr)f(O) + \frac{1}{e} h(O)-2\epsilon M \\
        &\geq \Bigl(1-\frac{1}{e}\Bigr) f(O) + \frac{1-\gamma_h}{e} f(O)-2\epsilon M
        = \Bigl(1-\frac{\gamma_h}{e}\Bigr)f(O)-2\epsilon M\\
        &\geq \Bigl(1-\frac{\gamma_h}{e}-O(\epsilon)\Bigr)f(O),
    \end{align*}
    where in the last inequality, we used the fact that $M = \max_{e \in E}\max\set{g(e),h(e)}\leq g(O)+h(O)$.
\end{proof}

In the rest of this section, we fix $f$ and its decomposition $f=g+h$, and the goal is a proof of Theorem~\ref{thm:simultaneous-approximation}.


\subsection{Continuous-time algorithm for simultaneous optimization}
First, we present a continuous-time version of our algorithm.
Although we cannot run this version in polynomial time, it is helpful to grasp the overall idea.

Our algorithm is based on the continuous greedy framework~\cite{Calinescu2011}.
Let $G:[0,1]^E \to \bbR_+$ be the multilinear extension of $g$ and $\bar h:[0,1]^E \to \bbR_+$ be the concave closure of $h$, and let $P_k$ be the convex hull of all characteristic vectors of subsets of size at most $k$.
We assume that we have value oracles of $G$, $\nabla G$, and $\bar{h}$.
We also assume that we know $g(O)$ and $h(O)$, where $O$ is the optimal solution to $\max\set{f(X):|X| \leq k}$.
We discuss how we can remove these assumptions in Section~\ref{subsec:discrete}.
The pseudocode of the continuous greedy algorithm is given in Algorithm~\ref{alg:cont-greedy}.

\begin{algorithm}[t]
    \begin{algorithmic}[1]
        \caption{Modified Continuous Greedy Algorithm}\label{alg:cont-greedy}
        \REQUIRE{$g(O),h(O) \in \bbR_+$, value oracles of $G,\bar{h}:[0,1]^E \to \bbR_+$, and a value oracle of $\nabla G:[0,1]^E \to \bbR_+^E$.}
        \STATE Set $\bx(0) := \bfzero \in \bbR^E$
        \STATE From $t=0$ to $t=1$, compute $\bx(t) \in [0,1]^E$ following the differential equation $\frac{d\bx}{dt} = \bv(t)$,
        where $\bv \in [0,1]^E$ is a vector satisfying
        \[
            \bv(t)^\top \nabla G(\bx) \geq g(O),\quad
            \bar{h}(\bv(t)) \geq h(O), \quad \text{ and } \quad
            \bv(t) \in P_k.
        \]
        \RETURN $\bx(1)$.
    \end{algorithmic}
\end{algorithm}

\begin{remark}
    One can find $\bv$ in Algorithm~\ref{alg:cont-greedy} in polynomial time.
    First, the problem of finding $\bv$ is a feasible problem of an LP.
    This problem is feasible because $\bfone_O \in \bbR^E$ is a feasible solution.
    Moreover, we have a separation oracle for $\bar{h}(\bv) \geq h(O)$.
    Hence, we can find a feasible solution using the ellipsoid method.
\end{remark}

The following lemma shows that Algorithm~\ref{alg:cont-greedy} provides the guarantee required in Theorem~\ref{thm:simultaneous-approximation}:
\begin{lemma}\label{lem:continuous}
    Let $\bx(t)$ be the sequence computed in Algorithm~\ref{alg:cont-greedy}.
    Then, we have $G(\bx(1)) \geq (1-1/e)g(O)$ and $\bar h(\bx(1)) \geq h(O)$.
\end{lemma}
\begin{proof}
    By adopting the standard analysis of the continuous greedy algorithm, we can see that $G(\bx(1)) \geq (1-1/e)g(O)$ (see, e.g.,~\cite{Calinescu2011}).
    On the other hand, since $\bar h$ is a concave function,
    \begin{align*}
        \bar{h}(\bx(1)) = \bar{h}\left(\int_0^1 \bv(t) dt \right) \geq \int_0^1 \bar{h}(\bv(t))dt \geq h(O),
    \end{align*}
    where the first inequality follows from Jensen's inequality.
\end{proof}

\subsection{Discrete-time algorithm for simultaneous optimization}\label{subsec:discrete}
Now, we describe the discrete-time version of Algorithm~\ref{alg:cont-greedy}, which can be run in polynomial time.
Again, we basically follow the argument in~\cite{Sviridenko2015} and explain the modifications to  Algorithm~\ref{alg:cont-greedy}.
We define $M= \max_{e \in E}\max\set{g(e),h(e)}$ as the maximum marginal gain of adding a single element.

\paragraph{Guessing $g(O)$ and $h(O)$:}
As we do not know the values of $g(O)$ and $h(O)$, we need to guess these values.
This is treated in exactly the same way as in~\cite{Sviridenko2015}.
First, we discuss how we guess the value of $g(O)$.
We note that $g(O)\leq  nM$ since $g$ is submodular.
Then, we divide the range $[0,nM]$ into polynomially many intervals, try the minimum value of each interval, and return the best solution found.
We now try all values of the form of either (i) $\alpha = i\epsilon M $ for $i = 0, 1, \dots, \lfloor 1/\epsilon \rfloor$ or (ii) $(1 + \epsilon/n)^i M$ for $i = 0, 1, \dots, \log_{1+\epsilon/n}n$.
If $g(O) < M$, then we have $g(O) - \epsilon M \leq \alpha \leq g(O)$ for some $\alpha$ considered in (i).
On the other hand, if $g(O) \geq M$, then we have $(1-\epsilon)g(O) \leq \alpha \leq g(O)$ for some $\alpha$ considered in (ii).
We can guess the value of $h(O)$ similarly.

The above exhaustive search on $g(O)$ and $h(O)$ increases the running time by a factor of $O((n \log n/\epsilon)^2)$, and we can assume that we know $\alpha$ and $\beta$ such that $(1-\epsilon)g(O) -\epsilon M \leq \alpha \leq g(O)$ and $(1-\epsilon)h(O) -\epsilon M \leq \beta \leq h(O)$.

\paragraph{Finding an appropriate direction $\bv$:}
Another technical issue is that we cannot compute the exact value of $\nabla G(\bx(t))$.
It is shown in~\cite{Calinescu2011} that for any constant $\epsilon > 0$, we can compute an approximation $\widetilde{\nabla G}(\bx)$ to $\nabla G(\bx)$ via random sampling such that $|\bv^\top \widetilde{\nabla G} - \bv^\top \nabla G| \leq \epsilon M$. 

\paragraph{Discretizing the time step:}
Lastly, we discretize the time step.
Let $\delta = \epsilon/n^2$, and for simplicity, we assume that $1/\delta$ is an integer.
In each time step $t$, we compute $\bv(t)$ as described above, and we set $\bx(t+\delta) \gets \bx(t) + \delta\bx(t)$.
For the $G$ part, following the argument in Lemma 3.3 of~\cite{Calinescu2011}, we have
\begin{align*}
    & G(\bx(t + \delta)) - G(\bx(t))
    \geq (1-n\delta)\delta\bv^\top \nabla G
    \geq (1-n\delta)\delta\bigl(g(O) - G(\bx(t)) - O(\epsilon M)\bigr) \\
    &\geq \delta \bigl(g(O) - G(\bx(t)) - O(\epsilon M) - (\epsilon/n)g(O)\bigr)
    =  \delta\bigl(g(O) - O(\epsilon M)\bigr).
\end{align*}
This yields $G(\bx(1)) \geq (1-1/e)g(O) - O(\epsilon M) \geq (1-1/e)g(O)-O(\epsilon M)$.

For the $\bar{h}$ part, by using the (discrete) Jensen inequality, we have
\begin{align*}
    \bar{h}(\bx(1)) = \bar{h}\left(\sum_{t = 0}^{1/\delta} \delta\bv(t) \right) \geq \sum_{t=0}^{1/\delta} \delta\bar{h}(\bv(t))  \geq h(O) - O(\epsilon M).
\end{align*}

Summarizing the discussion above, we have the following:
\begin{lemma}\label{lem:discrete}
  For any constant $\epsilon > 0$, there exists a polynomial-time algorithm that finds a vector $\bx \in P_k$ such that $G(\bx) \geq (1-1/e)g(O)-\epsilon M$ and $\bar h(\bx) \geq h(O) - \epsilon M$.
\end{lemma}

\subsection{Rounding}

In this section, we show the following:
\begin{lemma}\label{lem:rounding}
  Let $\bx \in P_k$ be a vector.
  Then, there exists a polynomial-time randomized rounding algorithm that outputs $X \subseteq E$ with $|X| \leq k$ such that $\E[g(X)] \geq G(\bx)$ and $\E[h(X)] \geq \bar{h}(\bx)$.
\end{lemma}
Theorem~\ref{thm:simultaneous-approximation} is obtained by combining Lemmas~\ref{lem:discrete} and~\ref{lem:rounding}.

First, we show a rounding method, provided that a convex combination of $\bar{h}(\bx)$ is known, in Section~\ref{subsubsec:rounding-cc-known}.
We will discuss how to compute a convex combination in Section~\ref{subsubsec:rounding-cc-compute}.

\subsubsection{Rounding when a convex combination of $\bar{h}$ is known}\label{subsubsec:rounding-cc-known}
Suppose that we have a convex combination of $\bar{h}(\bx)$; that is, we have $\lambda_1,\ldots,\lambda_m > 0$ and $X_1,\ldots,X_m \subseteq E$ of size $k$ with $\sum_i \lambda_i = 1$, $\sum_i \lambda_i \bfone_{X_i} = \bx$, and
\begin{align}\label{eq:convex-combination}
    \bar h(\bx) = \sum_{i=1}^m \lambda_i h(X_i).
\end{align}
Our rounding is a random process, and the condition $|X_i| = k\;(k \in [m])$ is always preserved throughout the process.
%

If $X_1= \cdots = X_m$, then we are done, and we output the set $X := X_1 (=\cdots = X_m)$, which is of size $k$.

Otherwise, find $1 \leq a < b \leq m$ such that $X_a \setminus X_b \neq \emptyset$ and $X_a \setminus X_b \neq \emptyset$.
Then, we fix $i \in X_a \setminus X_b$ and find $j \in X_b \setminus X_a$ \ such that
\begin{align*}
    h(X_a) + h(X_b) \leq h(X_a - i + j) + h(X_b + i - j),
\end{align*}
whose existence is guaranteed by the \M-concavity of $h$.
Then, we replace $X_a$ with $X_a - i + j$ with the probability $\lambda_b / (\lambda_a + \lambda_b)$, replace $X_b$ with $X_b + i - j$ with the complementary probability, and repeat this process again.

We now analyze the expected value of $h$ at the end of the rounding process.
\begin{lemma}\label{lem:rounding-h}
  Let $X \subseteq E$ be the output set.
  Then, we have $\E[h(X)] \geq \bar{h}(\bx)$.
\end{lemma}
\begin{proof}
    Suppose that we have a chosen $i \in X_a \setminus X_b$ and $j \in X_b \setminus X_a$ in an iteration of the process.
    Then, the expected change in the value of $\lambda_a h(X_a) + \lambda_b h(X_b)$ at this iteration is
    \[
        \frac{\lambda_a\lambda_b}{\lambda_a + \lambda_b}[h(X_a - i + j) - h(X_a) + h(X_b + i -j) - h(X_b)] \geq 0.
    \]
    Therefore, the expected value of $\sum_{i=1}^m \lambda_i h(X_i)$ is at least $\bar{h}(\bx)$.
    The lemma holds by induction on the number of iterations.
\end{proof}

Next, we analyze the expected value of $g$ at the end of the rounding process.
Let $\bx^t = \sum_i \lambda_i \bfone_{X_i}$ be the random vector after the $t$th exchange ($t = 0, 1, \dots$).
One can check that
(i) $\bx^0 = \bx$,
(ii) $\bx^{t+1} - \bx^t$ has at most one positive coordinate and at most one negative coordinate, and
(iii) $\E[\bx^{t+1} \mid \bx^t] = \bx^t$.
The following lemma establishes that such a random process preserves the value of the multilinear extension of a monotone submodular function.

\begin{lemma}[{\cite[Lemma~VI.2]{Chekuri2010}}]\label{lem:rounding-g}
    Let $\bx^t$ be a vector-valued random process satisfying the above conditions (i), (ii), and (iii).
    Then, for the multilinear extension $F$ of any monotone submodular function $f:2^E \to \bbR$, we have $\E [F(\bx^t)] \geq F(\bx)$ for $t = 0, 1, \dots$.
\end{lemma}

By Lemmas~\ref{lem:rounding-h} and~\ref{lem:rounding-g}, we have the following lemma.
\begin{lemma}\label{lem:rounding-cc-known}
    Let $X$ be the output of the above rounding process.
    Then, we have $\E [g(X)] \geq G(\bx)$ and $\E [h(X)] \geq \bar{h}(\bx)$.
\end{lemma}

Thus, we can round a fractional solution into an integral solution without losing the objective value (in the expectation).

\subsubsection{Computing a convex combination of $\bar{h}$}\label{subsubsec:rounding-cc-compute}

What remains to be shown is how to compute the convex combination in~\eqref{eq:convex-combination}.
Since $\bar{h}$ is a polyhedral concave function, $(\bx, \bar{h}(\bx))$ is contained in some face $F$ of the epigraph of $\bar{h}$.
Observe that it suffices to find a convex combination of $\bx$ in $P = \{ \by \in [0,1]^E : (\by, \bar{h}(\by)) \in F\}$.
To this end, we exploit the fact that $P$ forms a generalized polymatroid (g-polymatroid).

\begin{lemma}[{\cite[Proposition~6.23]{Murota2003},~\cite{Murota2000}}]
    Let $h$ be an \M-concave function and $\bp \in \bbR^E$. Then, $\argmin\{\bp^\top \bx - \bar h (\bx): \bx \in [0,1]^E  \}$ is an integral g-polymatroid.
\end{lemma}

Furthermore, we can focus on $P' = P \cap \{X \subseteq E : |X| = k\}$ since $\bx$ is in the convex hull of this set (note that $\bx(E) = k)$.
Note that if we can find a convex combination in $P'$, the obtained sets $X_1,\ldots,X_m$ must have a size of $k$.

Since any restriction of an integral $g$-polymatroid with the cardinality constraint yields a matroid base polytope (see, e.g.,~\cite{Murota2003,Fujishige2005}), we can use the algorithm of Cunningham~\cite{Cunningham1984} to find the desired convex combination.
To run this algorithm, we require a rank oracle of the corresponding matroid.
Such rank oracles can be implemented via the following characterization.

\begin{lemma}[folklore, {see~\cite[Theorem~1.9]{Murota2003}}]
    Given a matroid base polytope $B$, its associated rank function $r:2^E \to \bbZ_+$ can be written as
    \begin{align}
        r(X) = \max\{\bx(X) : \bx \in B\}.
    \end{align}
\end{lemma}

Therefore, we can obtain its rank function by linear optimization on $P'$.
More precisely, we can compute the rank of $X \subseteq E$ by the following convex programming:
\begin{align*}
    \text{maximize} & \quad \bx(X), \\
    \text{subject to} &\quad \bx(E) = k, \\
                   &\quad \bp^\top\bx - \bar{h}(\bx) \leq \alpha, \\
                   &\quad \bx \in [0,1]^E,
\end{align*}
where $\alpha := \min\{\bp^\top\by - \bar{h}(\by) : \by \in \bbR^E\}$.
This problem again can be solved by the ellipsoid method.


\section{Decomposition via \M-concave Quadratic Functions}\label{sec:decomposition}
In this section, we describe a general decomposition scheme that exploits \M-concave quadratic functions.

For a vector $\bd \in \bbR^E$, we define $\Diag(\bd) \in \bbR^{E \times E}$ as a diagonal matrix such that $\Diag(\bd)_{ii} = \bd(i)$ for every $i \in E$.
The following characterization of \M-concave quadratic set functions is known:
\begin{lemma}[\cite{Hirai2004}, also see~{\cite[Theorem~4.3]{Murota2016}}]\label{lem:mconcave-characterization}
    Let $A \in \bbR^{E \times E}$ be a symmetric matrix, and define $h:2^E \to \bbR$ as $h(X) = \frac{1}{2} \bfone_X^\top A \bfone_X = \frac{1}{2}\sum_{i \in X} a_{ii} + \sum_{i,j \in X: i \neq j} a_{ij}$.
    Then, $h(X)$ is \M-concave if and only if
    \begin{align}
        A_{ij} &\leq 0 \quad \text{for any distinct $i, j \in E$}, and \label{eq:quad-cond-1} \\
        A_{ij} &\leq \max\{A_{ik}, A_{jk}\}  \quad\text{for any distinct $i,j,k \in E$}. \label{eq:quad-ultrametric}
    \end{align}
    Furthermore, the above conditions are equivalent to the following:
    there exist a laminar family $\caL \subseteq 2^E$, $\lambda_L \leq 0$ ($L \in \caL$), and $\bd \in \bbR^E$ such that
    \[
        A = \sum_{L \in \caL} \lambda_L \bfone_L \bfone_L^\top + \Diag(\bd).
    \]
\end{lemma}
Note that the condition in~\eqref{eq:quad-ultrametric} says that $A_{ij}$ must form an \emph{ultrametric}.
The above condition is closely related to the concept of a discrete Hessian matrix:
\begin{definition}
    The \emph{discrete Hessian} $\Hess_f(X) \in \bbR^{E \times E}$ of a set function $f:2^E \to \bbR$ at a set $X \subseteq E$ is defined as
    \begin{align*}
        \Hess_f(X)_{ij} = f(X + i + j) + f(X) - f(X + i) - f(X + j). \quad (i, j \in E)
    \end{align*}
\end{definition}
We note that $\Hess_f(X)_{ii} = 0$ for any $i \in E$ and that $\Hess_f(X)_{ij} = 0$ when $i$ or $j$ belongs to $X$.
By definition, a set function $f:2^E \to \bbR$ is submodular if and only if $\Hess_f(X)_{ij} \leq 0$ for any $X \subseteq E$ and distinct $i,j \in E$.
A characterization of an \M-concave function in terms of a discrete Hessian is also known:
\begin{lemma}[{\cite[Theorem~4.3]{Murota2016}}]
    A set function $f: 2^E \to \bbR$ is \M-concave if and only if its discrete Hessian satisfies the conditions in~\eqref{eq:quad-cond-1} and~\eqref{eq:quad-ultrametric} at any $X \subseteq E$.
\end{lemma}

For a monotone submodular function $f:2^E \to \bbR_+$, consider a decomposition $f = g + h$, where $h(X) = \frac{1}{2} \bfone_X^\top A \bfone_X\;(X \subseteq E)$ for some matrix $A \in \bbR^{E \times E}$.
In the following, we derive a sufficient condition of $A$ that makes $g$ nonnegative monotone submodular and makes $h$ nonnegative \M-concave so that we can apply Theorem~\ref{thm:main}.

For the monotonicity of $g$, we must have $g(i \mid X) = f(i \mid X) - h(i \mid X) \geq 0$ for any $i \in E$ and $X \subseteq E - i$.
This yields
\begin{align}\label{eq:A-monotonicity}
    \sum_{k \in X}A_{ik} + \frac{1}{2}A_{ii} \leq f(i \mid X). \quad (i \in E, X \subseteq E -i)
\end{align}
We note that, if~\eqref{eq:A-monotonicity} was satisfied, $g$ becomes automatically nonnegative as $g(\emptyset)=0$.

For the submodularity of $g$, we must have $\Hess_g(X)_{ij} = \Hess_f(X)_{ij} - \Hess_h(X)_{ij} \leq 0$ for any distinct $i, j \in E$ and $X \subseteq E \setminus \{i, j\}$.
This yields
\begin{align}\label{eq:A-submodularity}
    A_{ij} \geq \Hess_f(X)_{ij}. \quad (i, j \in E, i \neq j, X \subseteq E \setminus \{i,j\})
\end{align}
Finally, for the nonnegativity of $h$, we must have
\begin{align}\label{eq:A-nonnegativity}
    \bfone_X^\top A \bfone_X \geq 0 \quad (X \subseteq E).
\end{align}

It is difficult to handle these conditions simultaneously.
Fortunately, we can show that we need to consider only the condition in~\eqref{eq:A-submodularity} and the nonpositivity constraint.

\begin{lemma}
    Suppose that for any distinct $i, j \in E$, we are given $A_{ij}$ with $\Hess_f(X)_{ij} \leq A_{ij} \leq 0$.
    Let us set
    $A_{ii} := 2f(i \mid E - i) - 2 \sum_{k \neq i} A_{ik}$ for $i \in E$.
    Then, the resulting symmetric matrix $A \in \bbR^{E \times E}$ satisfies the conditions in~\eqref{eq:A-monotonicity} and~\eqref{eq:A-nonnegativity}.
\end{lemma}
\begin{proof}
    Let us begin with the condition in~\eqref{eq:A-monotonicity}, which is equivalent to
    \begin{align*}
        A_{ii} \leq 2 \min_{X \subseteq E - i}\left[f(i \mid X) - \sum_{k \in X} A_{ik} \right]
    \end{align*}
    for every $i \in E$.

    We fix $i \in E$, and we will show that $X = E - i$ is a minimizer of $\min_{X \subseteq E - i}\left[f(i \mid X) - \sum_{k \in X} A_{ik} \right]$.
    Let us arbitrarily take $X \subseteq E - i$ and consider $j \not\in X$ ($j \neq i$).
    If we add $j$ to $X$, the change in the objective value is
    \begin{align*}
        f(i \mid X + j) - f(i \mid X) - A_{ij}
        &= f(X + i + j) - f(X + j) - f(X + i) + f(X) - A_{ij} \\
        &= \Hess_f(X)_{ij} - A_{ij} \leq 0.
    \end{align*}
    Thus, we can always add $j$ to $X$ without increasing the objective value.

    From the argument above, by setting $A_{ii} = 2f(i \mid E - i) - 2 \sum_{k \neq i} A_{ik}$ for every $i \in E$, we can satisfy the condition in~\eqref{eq:A-monotonicity}.
    Then, the condition in~\eqref{eq:A-nonnegativity} also follows since the resulting $A$ is diagonally dominant (note that $A_{ij} \leq 0$).
\end{proof}

To summarize, we have the following:
\begin{lemma}
    Suppose that a matrix $A \in \bbR^{E \times E}$ satisfies the conditions in~\eqref{eq:quad-cond-1},~\eqref{eq:quad-ultrametric}, and~\eqref{eq:A-submodularity}.
    Then, $h(X) = \frac{1}{2}\bfone_X^\top A \bfone_X$ is nonnegative \M-concave, and $g = f - h$ is nonnegative monotone submodular.
\end{lemma}
Trivially, we can take $A = 2\Diag(f(i \mid E - i) : i \in E)$, which yields $h(X) = \sum_{i \in X}f(i\mid E-i)$.
This decomposition corresponds to that via the standard curvature.
In the following, we discuss how we can find a nontrivial matrix $A$ satisfying the conditions in~\eqref{eq:quad-cond-1},~\eqref{eq:quad-ultrametric}, and~\eqref{eq:A-submodularity}.

\subsection{Ultrametric fitting problem}\label{subsec:ultrametric}

Assume that we can compute $H_{ij}$ such that $\max\limits_{X \subseteq E-i-j} \Hess_f(X)_{ij} \leq H_{ij} \leq 0$ for any distinct $i, j \in E$.
Once we are given such $H_{ij}$, finding a matrix $A$ satisfying the conditions in~\eqref{eq:quad-cond-1},~\eqref{eq:quad-ultrametric}, and~\eqref{eq:A-submodularity} boils down to the \emph{ultrametric fitting problem}~\cite{Farach1995}.
In this problem, given $H_{ij} \leq 0$ for distinct $i, j \in E$, we are to find $A_{ij} \in [H_{ij} , 0]$ for distinct $i, j \in E$ satisfying \eqref{eq:quad-ultrametric} that minimizes $\max_{i \neq j}|A_{ij} - H_{ij}|$.
Farach~et~al.~\cite{Farach1995} gave an $O(n^2)$-time algorithm for this problem.
Therefore, we can find a matrix $A$ satisfying the conditions in~\eqref{eq:quad-cond-1},~\eqref{eq:quad-ultrametric}, and~\eqref{eq:A-submodularity} that is close to $H_{ij}$.
Although this method is general, it is difficult to theoretically bound the curvature of the resulting decomposition.


\subsection{Coverage function}\label{subsec:coverage}
Let $G = (E, V; A)$ be a bipartite graph.
The coverage function $f:2^E \to \bbR$ associated with $G$ is defined as $f(X) := \sum_{v \in V} \min\{1, |\Gamma(v) \cap X|\}$, where $\Gamma(v) \subseteq E$ is the set of neighbors of $v \in V$.
It is well-known that $f$ is a monotone submodular function.

Now, we consider decomposing $f$ into a monotone submodular function and monotone \M-concave function.
For each $v \in V$, we define $f_v:2^E \to \bbR$ as $f_v(X) = \min\{1, |\Gamma(v) \cap X|\}$.
Then, for distinct $i, j \in E$ and $X \subseteq E - i -j$, $\Hess_{f_v}(X)_{ij}$ is equal to $-1$ if $v \not\in \Gamma(X)$ and $i,j \in \Gamma(v)$ and $0$ otherwise.
Thus, we have $\Hess_f(X)_{ij} = -  |\Gamma(i) \cap \Gamma(j) - \Gamma(X)|$.
Since $\Hess_f(X)_{ij}$ is obviously nondecreasing, we can compute the exact value of $\max_{X \subseteq E\setminus\{i,j\}} \Hess_f(X)_{ij} = \Hess_f(E-i-j)_{ij} = -c_{ij}$, where $c_{ij}= |\Gamma(i) \cap \Gamma(j) - \Gamma(E-i-j)|$.
Note that the quantity $c_{ij}$ is the number of $v \in V$ whose neighbors are exactly $i$ and $j$.
Then, we can run the algorithm presented in Section~\ref{subsec:ultrametric} with $H_{ij} = - |\Gamma(i) \cap \Gamma(j) - \Gamma(E-i-j)|$ to obtain a matrix $A$ satisfying the conditions in~\eqref{eq:quad-cond-1},~\eqref{eq:quad-ultrametric}, and~\eqref{eq:A-submodularity}.
We note that the standard curvature can only handle vertices in $V$ that have only one neighbor in $E$, whereas our argument can handle vertices in $V$ that have two neighbors in $E$.





\section{Applications}\label{sec:applications}
In this section, we apply Theorem~\ref{thm:main} to obtain better approximation guarantees for the facility location problem and the sum of weighted matroid rank functions.
In these applications, we need to use the specific structures of these problems.



\subsection{Facility location}
In the \emph{facility location problem}, there are a set $I$ of customers and a set $E$ of possible locations of facilities.
Each customer $i \in I$ has a revenue $w_{ij} \geq 0$ for the facility $j \in E$.
We assume that customers will select the available facility of maximum revenue.
The task is to select a set $X \subseteq E$ of cardinality $k$ that maximizes
\[
    f(X) = \sum_{i \in I} \max_{j \in X} w_{ij},
\]
where we conventionally define $f(\emptyset)=0$.
Evidently, $f$ is a nonnegative monotone submodular function.

We can decompose $f$ into $g$ and $h$ as follows.
For $i \in I$, we denote $w_{i, \min}:= \min_{j \in E}w_{ij}$ and let $\bar{w}_{ij} := w_{ij} - w_{i,\min}$ ($i \in I$, $j \in E$).
Then, we can rewrite $f$ as
\[
    f(X) = \sum_{i \in I} \max_{j \in X} \bar{w}_{ij} + \left(\sum_{i \in I} w_{i,\min}\right)[X \neq \emptyset],
\]
where $[X \neq \emptyset]$ is the indicator function of the nonemptiness of $X$.
Let $\tilde{f}(X) = \sum_{i \in I} \max_{j \in X} \bar{w}_{ij}$ be the first term.
We further subtract a modular function $\ell(X) = \sum_{j \in X} \tilde{f}(j \mid E-j)$ from $\tilde{f}$.
Note that since $f(j \mid E - j) = [w_{ij} - \max_{k \neq j}w_{ij}]^+ = [\bar{w}_{ij} - \max_{k \neq j}\bar{w}_{ij}]^+ = \tilde{f}(j \mid E-j)$ for $j \in E$, this modular term is exactly the same one in the curvature decomposition for the original function $f$.
Then, we define $g(X) := \tilde{f}(X) - \ell(X)$ and $h(X):= \ell(X) + \left(\sum_{i \in I} w_{i,\min}\right)[X \neq \emptyset]$.
One can easily check that $g$ is monotone submodular and $h$ is \M-concave.

We can show that the generalized curvature $\gamma_h$ is no more than the standard curvature $c$.
\begin{lemma}\label{lem:h-curvature-facility-location}
    $\gamma_h \leq c - \frac{\sum_{i \in I}w_{i,\min}}{\sum_{i \in I}w_{i, \max}}$, where $w_{i, \max}:= \max_{j \in E}w_{ij}$ ($i \in E$).
\end{lemma}
\begin{proof}
    For any nonempty $X \subseteq E$, we have
    \begin{align*}
        \frac{h(X)}{f(X)}
        = \frac{\ell(X)}{f(X)} + \frac{\sum_{i \in I} w_{i,\min}}{f(X)}
        \geq 1 - c + \frac{\sum_{i \in I} w_{i,\min}}{f(E)}
        = 1 - c +  \frac{\sum_{i \in I}w_{i,\min}}{\sum_{i \in I}w_{i, \max}},
    \end{align*}
    where the inequality follows from the definition of the curvature and the monotonicity of $f$.
    Therefore, $\gamma_h = 1 - \min_{X \subseteq E}\frac{h(X)}{f(X)} \leq c - \frac{\sum_{i \in I}w_{i,\min}}{\sum_{i \in I}w_{i, \max}}$.
\end{proof}

\subsection{Sums of weighted matroid rank functions}
A weighted matroid rank function is an \M-concave function (see Example~\ref{ex:weighted-rank}), and its sum is a monotone submodular function.
A sum of weighted matroid rank functions includes the coverage function, the facility location function, and many others, but not all monotone submodular functions.
Here, we show that if the objective function is a sum of weighted matroid rank functions, we can obtain an improved approximation guarantee.

Let $f(X) = \sum_{i \in I} f_i(X)$, where $f_i(X) = \max\{\bw_i(J): J \subseteq X, J \in \caI_i\}$ for some vector $\bw_i \in \bbR_+^E$, and $\caI_i$ is the independent set family of a matroid ($i \in I$).
For $i \in I$, let $w_{i, \min} := \min_{e\in E}\bw_i(e)$.
Let $\tilde{f}(X) = \sum_{i \in I} \tilde{g}_i(X)$, where $\tilde{g}_i(X)$ is the weighted matroid rank function with the reduced weight $\bw_i - w_{i, \min}\bfone$ for $i \in I$.
Define $\ell(X) = \sum_{j \in X} \tilde{f}(j \mid E - j) = \sum_{j \in X} f(j \mid E-j)$, where the second equality follows from the fact that all matroid bases have the same cardinality.
Finally, let us define $g(X) := \tilde{f}(X) - \ell(X)$ and $h(X) := \ell(X) + (\sum_{i \in I} w_{i, \min})[X \neq \emptyset]$.
The proof of the following lemma is similar to that of Lemma~\ref{lem:h-curvature-facility-location}.
\begin{lemma}
    $\gamma_h \leq c - \frac{\sum_{i \in I}w_{i,\min}}{\sum_{i \in I}W_i}$, where $c$ is the standard curvature, and $W_i := f_i(E)$ is the maximum weight of the bases in $(E, \caI_i)$ for each $i$.
\end{lemma}

In~\cite{Lin2012}, Lin and Bilmes proposed a linear mixture model of simple monotone submodular functions.
One of the most general classes in their model is a linear mixture of weighted matroid rank functions.
Suppose that a function $f:2^E \to \bbR_+$ is obtained by inference for this linear mixture model.
Then, we can write $f(X) = \sum_{i \in I}\alpha_i f_i(X)$, where $\alpha_i > 0$ is a mixture coefficient, and $f_i(X)$ is a weighted matroid rank function ($i \in I$).
If some coefficient $\alpha_{i^*}$ is dominant, one can consider the following straightforward decomposition: $g(X) = \sum_{i \neq i^*}\alpha_i f_i(X)$ and $h(X) = \alpha_{i^*}f_{i^*}(X)$. 
Then, we can expect the resulting curvature to be small, although the actual value depends on the form of $f_i\;(i \in I)$.

\subsection*{Acknowledgment}
The authors thank Yuni Iwamasa for pointing out a reference~\cite{Farach1995} on ultrametric fitting.

\bibliographystyle{IEEEtranS}
\bibliography{mconcave}
\end{document}